\documentclass[journal,onecolumn]{IEEE_lsens}
\usepackage{textcomp}
\usepackage{amsmath}
\usepackage{graphicx}
\usepackage{subcaption}
\usepackage{amsthm}
\usepackage[utf8]{inputenc}
\usepackage[english]{babel}
\usepackage[noadjust]{cite}
\newtheorem{theorem}{Theorem}
\usepackage[T1]{fontenc} 

\usepackage{amsmath}
\usepackage[cmintegrals]{newtxmath}
\usepackage{bm}
\usepackage{array}
\usepackage{url}

\begin{document}

\IEEELSENSarticlesubject{}

\title{Covariance Matching based robust Adaptive Cubature Kalman Filter}

\author{\IEEEauthorblockN{Mundla Narasimhappa\IEEEauthorrefmark{1}\IEEEauthorieeemembermark{1} and Sesham Srinu\IEEEauthorrefmark{2}}% <-this % stops a space
\IEEEauthorblockA{\IEEEauthorrefmark{1}Department of Mechanical Engineering,
University of Surrey, Surrey, GU2 7XH, UK\\
\IEEEauthorrefmark{2}Department of Electrical Engineering, University of Nimibia, Nimibia, South\\
}}

% \thanks{Corresponding author: Mundla Narasimhappa (e-mail:n.mundla@surrey.ac.uk.)}}

\IEEEtitleabstractindextext{%
\begin{abstract}
This letter explores covariance matching based adaptive robust cubature Kalman filter (CMRACKF). In this method,  the  innovation sequence is used to determine the covariance matrix of measurement noise that can overcome the limitation of conventional CKF. In the proposed algorithm, weights are  adaptively adjust and used for updating the measurement noise covariance matrices online. It can also enhance the adaptive capability of the ACKF.  The simulation results are illustrated to evaluate the performance of the proposed algorithm.
\end{abstract}

}

\maketitle

\section{Introduction}

\IEEEPARstart{K}alman filter (KF) has been widely used state estimator in several fields; in navigation, target tracking \cite{grewal2014kalman},  robotics, control and signal processing\cite{cui2016adaptive}. In practice, the system  and measurement model of the KF are not known exactly. Moreover, the system becomes nonlinear and their noise models are vary with time \cite{grewal2014kalman}. Over years, several variants of non-linear estimation  methods; extended KF (EKF),  Unscented Kalman filter (UKF), and cubature Kalman filter (CKF) have been developed  for estimating state of a nonlinear system \cite{wan2000unscented, arasaratnam2009cubature}. The EKF performance is limited by the calculation of Jacobin errors. To improve the performance and  by addressing the EKF drawback, UKF has been developed based on Unscented Transform \cite{wan2000unscented}. In the UKF, a set of sigma points has used to approximate the mean and covariance of the state vectors. However, the accuracy  of UKF  has limited for higher order system \cite{meng2016covariance, cui2016adaptive}. CKF has been proposed to eliminate the limitation of UKF  and used for solving the higher order systems. The CKF has developed by  set of cubature points and used  to approximate the state vector in terms of posteriori mean and covariance. The CKF and UKF have the same accuracy, however better than the EKF filter \cite{arasaratnam2009cubature}.

In the CKF, a priori knowledge of the system and  noise models are usually unknown in real-time then  CKF becomes sub-optimal. Moreover, the uncertainties in the system and measurement models  may leads to large errors. Hence, CKF shows divergence behaviour under these conditions \cite{gao2015sage}. Recent past, many authors have been focusing on adaptive CKF and being addressed limitation on CKF that can applied in nonlinear system analysis with additive Gaussian noise \cite{meng2016covariance}. The innovation or residual based adaptive methods have been developed and followed by using covariance matching, Bayesian approach, multiple model estimation (MME) and  Maximum likelihood method \cite{mohamed}. In covariance matching based AKF \cite{mohamed}, in which window average method is used to estimate the noise covariances matrices \cite{xia2015state}. However, for improving the practicality of the adaptive CKF, robust methods are better by adaptively updating  the noise statistics online \cite{zhenbing2018adaptive, narasimhappa2016fiber}.

To the best of authors knowledge's, there is a limited work on robust adaptive cubature Kalman filter. This letter presents an improved covariance matching based robust adaptive cubature Kalman filter (CMRACKF) for addressing the outliers in the measurements. In the proposed algorithm,  moving average method (MAM) is used to adapt the noise statistics of innovation vector and the weights on each window are adjusted. The ability of the CMRACKF algorithm is to estimate and update the noise statistics online.

 The rest of the paper was organized as follows; The description of the traditional Cubature Kalman filter algorithm is presented in Section II. Section III provides the briefly on ACKF algorithm and then the proposed CMRACKF algorithm in Section IV. Simulation examples along with performance  analysis of the proposed algorithm is  given in Section IV. Section V presents the conclusions of the paper.

\section{Cubature Kalman filtering}
In nonlinear state estimation, Cubature Kalman Filter (CKF) is the better method has been widely used algorithm for solving the higher dimensions of the  stochastic systems. In the CKF, cubature points are required to approximate the state vector mean and error covariance followed by the spherical radial cubature criterion \cite{arasaratnam2009cubature}.  

Let us consider a discrete time stochastic nonlinear dynamic system and measurement equations:

\begin{equation}\label{eqn4.1}
 {\bf x_{k}}={\bf f({x_{k-1}},u_{k-1})}+{\bf w_{k-1}}
 \end{equation}
\begin{equation}\label{eqn4.2}
 {\bf z}_{k}={\bf h({\bf x_{k}})}+{\bf v_{k}}
 \end{equation}

where, $\bf x_{k}  \in R^{n}$ is the state vector, ${\bf u_{k}}  \bf \in R^{r}$ is the input vector,  ${\bf z_{k}} \bf \in R^{m}$ is the measurement vector at time ${k}$. ${\bf f({x}_{k-1})}, {\bf h({k})}$ are the nonlinear system dynamic  and measurement functions. The process and measurement noise are assumed to be white Gaussian noise with zero mean and finite variance, represented as ${\bf w_k}=N(0, {\bf Q_{k}})$ and  ${\bf v_k}=N(0, {\bf R_{k}})$, respectively. The step-wise implementation of CKF algorithm is follows \cite{arasaratnam2009cubature}

The detailed algorithm of the CKF is given as follows.\\
Step 1: Initialize  the state estimation ${\bf \hat x}_{0}$ and error covariance matrix $ {\bf \hat P}_{0} $  as
    \begin{equation}
      \label{eqn4.3}
    \begin{cases}
    {\bf \hat x_{0}}=E[{\bf  x_{0}}]\\
     {\bf \hat P_{0}}=E[({\bf x_{0}}-{\bf \hat x_{0}})({\bf x_{0}}-{\bf \hat x_{0}})^T]
    \end{cases}
    \end{equation}

A set of 2L cubature points are given by set $ [ {({\bf \chi_k})_i}  \quad { W_{i}}^{m} ] $, where ${({\bf \chi_k})_i}$ is the i-th cubature point and corresponding weights are represented as 
 
 \begin{equation}
  \label{eq:t}
  \begin{aligned}
     {{\bf \chi}_i} & =\sqrt{\bf L}[1]_{i},\\        
    {\bf W_{i}}^{m} &=\frac{1}{2L}, &i=1,2..2L\\
  \end{aligned}
\end{equation}
 
 where, [1](i)  denotes its i-th column vector of the the identity matrix.
 The steps involved in the predicted (time-update) and the measurement-update of the CKF are summarised in below.
 
Prediction:\\
(1) factorise and evaluate the cubature points:

 \begin{equation}
  \label{eq:t}
  \begin{aligned}
  {\bf \hat P_{k-1}} &=S_{k-1}S_{k-1}^T,\\
  {({\bf \chi_k})_i} &=S_{k}  {{\bf \xi}_i}+ {\bf \hat x_{k}}, &i=1,2..2L\\
  \end{aligned}
\end{equation}

(2) Propagate each sigma points through the  nonlinear system  as
\begin{equation} 
{({\bf X_{k-1}})_i}= {\bf f}({{({\bf \chi_k})_i}, \bf u_{k}}),  \quad   i=1,2 ...... 2L
 \label{stateSpaceForm1}
 \end{equation}
 
(3) Evaluate the predicted  state ${\bf \hat x_{k}}^{-}$ and state error covariance ${\bf \hat P_{k}}^{-}$  based on  the  transformed sigma points  as

 \begin{equation}
  \label{eq:t}
  \begin{aligned}
{\bf \hat x_{k-1}}^{-} & =\frac{1}{2L}\sum\limits_{i=1}^{2L}{({\bf X_{k-1}})_i},\\
{\bf \hat P_{k-1}}^{-} & =\frac{1}{2L}\sum\limits_{i=1}^{2L}{({\bf X_{k-1}})_i} {({\bf X_{k-1}})_i^\top} -{\bf \hat x_{k-1}}{\bf \hat x_{k-1}}^{{-}^\top}+{\bf Q_{k-1}}\\
  \end{aligned}
\end{equation}

Measurement update:\\
(1) factorise and evaluate the new cubature points:
 \begin{equation}
  \label{eq:t}
  \begin{aligned}
  {\bf \hat P_{k-1}} &=S_{k-1}S_{k-1}^T,\\
  {({\bf Z_{k-1}})_i} &=S_{k-1}  {{\bf \xi}_i}+ {\bf \hat x_{k-1}}, &i=1,2..2L\\
  \end{aligned}
\end{equation}

(2) Propagate new sigma points through the  nonlinear system  as
\begin{equation} 
{({\bf z_{k-1}})_i}= {\bf h}({{({\bf Z_{k-1}})_i}}),  \quad   i=1,2 ...... 2L
 \label{stateSpaceForm1}
 \end{equation}
 
(3) Evaluate the predicted measurements ${\bf \hat z_{k}}^{-}$ based on the  new sigma points  as

 \begin{equation}
  \label{eq:t}
  \begin{aligned}
{\bf \hat z_{k-1}}^{-} & =\frac{1}{2L}\sum\limits_{i=1}^{2L}{({\bf Z_{k-1}})_i},\\
\end{aligned}
\end{equation}

(4) calculate the  the cross and auto covariance of state and measurement values of ${\bf P_{xz,k}} $  and  ${\bf P_{zz,k}}$ are 
\begin{equation} 
  \label{eq:t}
  \begin{aligned}
{\bf P_{xz,k-1}}&=\frac{1}{2L}\sum\limits_{i=1}^{2L} {({\bf X_k})_i}  {({\bf Z_k})_i^\top}-{\bf \hat x_{k-1}}^{-} {\bf \hat z_{k-1}}^{{-}^\top}\\
{\bf P_{zz,k-1}}&=\sum\limits_{i=1}^{2L} {({\bf Z_{k-1}})_i} {({\bf Z_{k-1}})_i^\top}-{\bf \hat z_{k-1}}^{-} {\bf \hat z_{k-1}}^{{-}^\top}+{\bf R_{k-1}}
\end{aligned}
 \end{equation}
 (5) evaluate the Kalman gain and the updated stat and error covariance are
\begin{equation} 
  \label{eq:t}
  \begin{aligned}
{\bf K_{k}}&={\bf P_{xz,k-1}}{\bf P_{zz,k-1}^{-1}},\\
{\bf \hat x_{k}}& ={\bf \hat x_{k-1}^{-}}+{\bf K_{k}}({\bf z_{k}}-{\bf \hat z_{k-1}^{-}})\\
{\bf \hat P_{k}} &={\bf \hat P_{k-1}}^{-}-{\bf K_{k}}{\bf P_{zz,k-1}}{\bf K_{k}}^\top.
\end{aligned}
 \end{equation}
 
where,  (${\bf \upsilon_k}={\bf z_{k}}-{\bf \hat z_{k-1}^{-}}$) is the innovation sequence. ${\bf \hat P_{k}}$, is  the posterior state estimate of state. More detailed explanation of CKF can be found in  \cite{arasaratnam2009cubature}
\section{Adaptive Cubature Kalman Filter}
To improve the optimality and solving divergence issue of CKF,  adaptive Kalman filters \cite {mohamed} have been developed for updating the ${\bf Q_{k}} $ and ${\bf R_{k}}$ online. An innovation based adaptive estimation (IAE) ACKF is the one widely used estimation method \cite {gao2015sage}.

\subsection{Innovation based adaptive estimation Adaptive cubature Kalman filtering (IAE-ACKF)}

To address the CKF limitation and solving the divergence solution, the ACKF has briefly presented \cite{xia2015state,narasimhappa2016fiber}. In the ACKF, the innovation sequence ($ {\bf \upsilon_{k}}={\bf z_{k}}-{ \bf h} {\bf \hat {x}^{-}_{k}}$) is  calculated between measured and predicted measurements difference. By taking covariance, i.e., $E({\bf \upsilon_{k}}{\bf \upsilon_{k}^\top})$, then theoretical covariance matrix of $ { \bf C_{\upsilon k}}$ is
\begin{equation}
 { \bf C_{\upsilon k}}={\bf h} {\bf \hat P^{-}_{k}}{\bf h^T}+{\bf {R}_{k}}
  \label{stateSpaceForm1}
  \end{equation}

The window average method and followed by the covariance matching principle in ACKF \cite{mohamed}, is used to estimate the covariance matrix of innovation sequence as
\begin{equation}
 {\bf \hat C_{\upsilon k}}=\frac {1}{N_{w}}\sum_{j=j0}^{k}{\bf \upsilon_{j}} {\bf \upsilon_{j}}^T
  % \hat{x}	\hat{x}
  \label{stateSpaceForm1}
  \end{equation}
 where,  $j_{0}={k}-{N_w}+1 $ is the first epoch.  ${\bf \upsilon_{j}}$ and  ${N_w}$  are  the innovation sequence and  window size respectively. If the window size is too small, the estimation of measurement noise covariance can be noisy  \cite{meng2016covariance}. 

\section{Proposed algorithm: CMRACKF}
 In the proposed method, the measurement noise covariance is matrix is estimated; (i) covariance matching principle (ii) estimated  measurements variance (see in equation (16)). At each epoch, the individual weights contribution to data samples of measurements and state estimator calculation by the estimated variance of measurements. In the CMRACKF, an adaptive tuning parameter also is used to control the disturbances due to the effect of the residual error measurements. 

\subsubsection{Adaption of measurement noise covariance matrix}
In the CMRACKF algorithm, the estimated  measurement noise covariance matrix (${\bf \hat R_{k}}$) is updated according to the innovation sequence as
\begin{equation}\label{eqn33}
{\bf \hat R}_{k-j}^{*}=\frac {1}{N_w}\sum_{j=j0}^{k} {\omega({\bf \sigma^2_{k-j}})}{\bf \upsilon_{j}}{\bf \upsilon_{j}}^T-{\bf h}{\bf \hat P_{k}^{-}}{\bf h}^T
  % \hat{x}	\hat{x}k
%   \label{stateSpaceForm1}
  \end{equation}

 where, $ {\omega({\bf \hat \sigma^2_{k-j}})}$ is the weighted function of the measurement error is calculated as follows; 

\subsubsection{Calculation of Weight}
The estimated variance and corresponding weights of measurement are defined as  \cite{narasimhappa2016fiber}

\begin{equation}\label{eqn28}
\begin{aligned}
 {\bf \hat\sigma^2_{k-j}}&=\frac{{\bf \upsilon_{k}^T {\bf R_{k}{\bf \upsilon_{k}}}}}{\bf r_{k}},\\
  {\bf \omega({\bf \hat\sigma^2_{k-j}})} & =\frac{(\frac{1}{{\bf \hat\sigma^2_{k-j}}})}{({\sum_{j=1}^{N_w}{\frac{1}{{\bf \hat\sigma^2_{k-j}}}}})}
 \end{aligned}
  \end{equation}
  
subject to the sum of the weights  are equal to 1, i.e; $ \sum_{j=1}^{N_{w}}{\bf \omega({\bf \hat\sigma^2_{k-j}})}=1$. Where, ${\bf r_{k}}$ denotes the number of measurements at each epoch $k$ \cite{gao2015sage}. In the CMRACKF algorithm, the estimated covariance of innovation sequence is updated by 
\begin{equation}
 {\bf \hat C_{\bf \upsilon k}}=\frac {1}{N_{w}}\sum_{j=j0}^{k}{\omega({\hat \sigma^2_{k-j}})}{\bf \upsilon_{j}} {\bf \upsilon_{j}^\top}
  \label{stateSpaceForm1}
  \end{equation}

where  ${\omega(.)}$,  $\bf \upsilon_{j}$ and  ${\bf \hat C_{\bf \upsilon k}} $ are weight function,  innovation sequence and its covariance matrix, respectively.\\

%
%\subsection{Noise Statistic Estimator Based on Moving Window Method}
\begin{theorem}
Suppose the measurement and system noise statistics are two noise control parameters are very small in considered  window size  of  ${N_w}$. Then, a novel noise statistic estimator can be developed and derived as

% \begin{equation} 
% {\bf \hat R}_{k-j}^{*}=\frac{1}{{N_w}}\sum\limits_{j=0}^{N} {\omega({\bf \sigma^2_{k-j}})} [{\bf \upsilon_{j}}  {\bf \upsilon_{j}}^\top-{\bf h_{k-j}} {\bf \hat P_{k-j}}  {\bf h_{k-j}^\top}  ] 
%  \label{stateSpaceForm1}
%  \end{equation}.
 
 \begin{equation}\label{eqn28}
{\bf \hat R_{k-j}^{*}}=\frac{1}{N_w}\sum\limits_{j=0}{^{N}} {\omega({\bf \sigma^2_{k-j}})} {\bf \upsilon_{j}}  {\bf \upsilon_{j}^\top}-{\bf h_{k-j}} {\bf \hat P_{k-j}}  {\bf h_{k-j}^\top}  
\end{equation}

  \begin{gather}
{\bf \hat Q}_{k-j}^{*} =\frac{1}{{N_w}}[\sum\limits_{j=0}^{N} {\omega({\bf \sigma^2_{k-j}})} {\bf \upsilon_{j}}  {\bf \upsilon_{j}}^\top+\sum\limits_{j=0}^{N} {\omega({\bf \sigma^2_{k-j}})} {\bf \eta_{j}}  {\bf \eta_{j}}^\top ]\\ \notag -{\bf h_{k-j}} [\frac{1}{2L}\sum\limits_{i=1}^{2L}{({\bf X_{k-1}})_i} {({\bf X_{k-1}})_i^\top} -{\bf \hat x_{k-1}}{\bf \hat x_{k-1}}^{{-}^\top}+  {\bf \hat P_{k-j}} ] {\bf h_{k-j}^\top}
\end{gather}
 
 \end{theorem}

\begin{proof} Let us assume that the window width is ${N_w}$ and there are ${N_w}$ measurements within $t_{k-N}$ to $  t_{k}$. Suppose the measurement noise statistics is very small variations within window width.
The innovation is described as
 \begin{equation}
 {\bf \upsilon_{k}}={\bf z_{k}}-{\bf \hat z_{k}^{-}}
  \label{stateSpaceForm1}
  \end{equation}
and substituting the measurement equation in (2) into (21), we have

  \begin{equation}
 {\bf \upsilon_{k}}={\bf h(.)}({\bf {x}_{k}}-{\bf \hat {x}^{-}_{k}})+{\bf v_k}
  \label{stateSpaceForm1}
  \end{equation}
Define the predicted and estimated state error are 
 \begin{equation}
  \label{eq:t}
  \begin{aligned}
{\bf \Delta \hat {x}^{-}_{k}}&={\bf {x}_{k}}-{\bf \hat {x}^{-}_{k}},  {\bf \hat P_{k}^{-}}=E[{\bf \Delta \hat {x}^{-}_{k}} {\bf \Delta \hat {x}^{-}_{k}}^\top]\\
{\bf \Delta \hat {x}_{k}}&={\bf {x}_{k}}-{\bf \hat {x}_{k}},  {\bf \hat P_{k}}=E[{\bf \Delta \hat {x}_{k}} {\bf \Delta \hat {x}_{k}}^\top]\\
  \end{aligned}
\end{equation}  
By taking the sample mean and the covariances for a limited number of sample of the innovation sequence is   
 \begin{equation}
  \label{eq:t}
  \begin{aligned}
{\bf\bar \upsilon_{k}}&=\frac{1}{{N_w}}\sum\limits_{j=1}^{{N_w}} {\bf \upsilon_{k-j}},\\
E[{\bf \upsilon_{k}} {\bf \upsilon_{k}}^\top]&= \frac{1}{{N_w}-1}\sum\limits_{j=1}^{{N_w}} {\bf \omega({\bf \hat\sigma^2_{k-j}})} ({\bf \upsilon_{k-j}}-{\bf\bar \upsilon_{k}}) ({\bf \upsilon_{k-j}}-{\bf\bar \upsilon_{k}}^\top)
  \end{aligned}
\end{equation}

From equation (22), we can apply the weighted sample covariances 

 \begin{equation}
  \label{eq:t}
  \begin{aligned}
E[{\bf \upsilon_{k}} {\bf \upsilon_{k}}^\top]&=E[{\bf h}({\bf {x}_{k}}-{\bf \hat {x}^{-}_{k}})+{\bf v_k}] [{\bf  h}({\bf {x}_{k}}-{\bf \hat {x}^{-}_{k}})+{\bf v_k})]^\top\\
&= E[{\bf h}({\bf \Delta \hat {x}^{-}_{k}})+{\bf v_k}] [{\bf  h}({\bf \Delta \hat {x}^{-}_{k}})+{\bf v_k})]^\top\\
&= {\bf  h}E[({\bf \Delta \hat {x}^{-}_{k}}) ({\bf \Delta \hat {x}^{-}_{k}})^\top] {\bf  h^\top} + E[ {\bf v_k} {\bf v}^\top_k] \\
&= {\bf h}{\bf \hat P_{k}} {\bf  h^\top} + {\bf  R_{k}} \\
&={\bf P_{zz,k}}
  \end{aligned}
\end{equation}

% From equation (12),
Thus, auto-covariance of innovation covariance is equal to 
 \begin{equation}
 \begin{split}
{\bf P_{zz,k}}&=E[{\bf \upsilon_{k}} {\bf \upsilon_{k}^\top}]\\
&=\sum\limits_{i=1}^{2L} {({\bf Z_{k-1}})_i} {({\bf Z_{k-1}})_i^\top}-{\bf \hat z_{k-1}}^{-} {\bf \hat z_{k-1}}^{{-}^\top}+{\bf R_{k}^{*}}\\
{\bf \hat R}_{k}^{*}&=E [{\bf \upsilon_{k}} {\bf \upsilon_{k}^\top}]-{\bf h_{k-j}}{\bf P_{k-j}} {\bf h_{k-j}^\top}
  \label{stateSpaceForm1}
  \end{split}
  \end{equation}

After applying  MAM, uses sample sequence of innovation as 
   \begin{equation}
{\bf \hat R}_{k-j}^{*}=\frac {1}{N_w}\sum_{j=j0}^{k} {\omega({\bf \sigma^2_{k-j}})}{\bf \upsilon_{j}}{\bf \upsilon_{j}^\top}-{\bf h_{k-j}}{\bf P_{k-j}} {\bf h_{k-j}^\top}
  \label{stateSpaceForm1}
  \end{equation}

% By using equation (25), the estimation of ${\bf R_{k}}$ can be obtained

%   \begin{equation}
% {\bf \hat R}_{k-j}^{*}=\frac{1}{N_w}\sum\limits_{j=0}^{N_w}[E[{\bf \varepsilon_{k}} {\bf \varepsilon_{k}}^\top]-{\bf h_{k-j}}{\bf P_{k-j}} {\bf h_{k-j}}^\top]
%   \label{stateSpaceForm1}
%   \end{equation}

% \begin{equation}
% \begin{split}
%  {\bf \hat  Q_{k-1}}&=\frac{1}{N}\sum\limits_{j=0}^{N}  {\bf \hat  Q_{k-1-j}}\\
%  & =  \begin{aligned} {\bf \hat P_{k-j}}+ {\bf K_{k-j}} E[{\bf \varepsilon_{k-j}} {\bf \varepsilon_{k-j}}^\top] {\bf K_{k-j}}^\top \\ -\sum\limits_{j=0}^{2L}{\bf W_{i}^{c}}[{({\bf X}_{k-j})_i} -{\bf \hat x_{k-j}}^{-}][{({\bf X}_{k-j})_i}-{\bf \hat x_{k-j}}^{-}]^\top \end{aligned}
%   \label{stateSpaceForm1}
%   \end{split}
%   \end{equation}

%  \begin{equation}
% E[{\bf \varepsilon_{k-j}} {\bf \varepsilon_{k-j}}^\top]={\bf \varepsilon_{k-j}} {\bf \varepsilon_{k-j}}^\top
%   \label{stateSpaceForm1}
%   \end{equation}

% substituting equation (25) into equation (24), the estimation of ${\bf Q_{k-1}}$ can be obtained
% \begin{equation}
%  \begin{split}
%   {\bf \hat Q_{k-1-j}}=\frac{1}{N}[{\bf \hat P_{k-j}}+ {\bf K_{k-j}} E[{\bf \varepsilon_{k-j}} {\bf \varepsilon_{k-j}}^\top] {\bf K_{k-j}}^\top\\-\sum\limits_{j=0}^{2L}{\bf W_{i}^{c}}[{({\bf X}_{k-j})_i} -{\bf \hat x_{k-j}}^{-}][{({\bf X}_{k-j})_i}-{\bf \hat x_{k-j}}^{-}]^\top]
%   \label{stateSpaceForm1}
%   \end{split}
%   \end{equation}

The equation (27) completes the proof.

By taking the sample mean and the covariances for a limited number of sample of the innovation sequence is   
 \begin{equation}
  \label{eq:t}
  \begin{aligned}
{\bf\bar \upsilon_{k}}&=\frac{1}{{N_w}}\sum\limits_{j=1}^{{N_w}} {\bf \upsilon_{k-j}},\\
E[{\bf \upsilon_{k}} {\bf \upsilon_{k}}^\top]&= \frac{1}{{N_w}-1}\sum\limits_{j=1}^{{N_w}} {\bf \omega({\bf \hat\sigma^2_{k-j}})} ({\bf \upsilon_{k-j}}-{\bf\bar \upsilon_{k}}) ({\bf \upsilon_{k-j}}-{\bf\bar \upsilon_{k}}^\top)
  \end{aligned}
\end{equation}

From equation (22), we can apply the weighted sample covariances 

 \begin{equation}
  \label{eq:t}
  \begin{aligned}
E[{\bf \upsilon_{k}} {\bf \upsilon_{k}}^\top]&=E[{\bf h}({\bf {x}_{k}}-{\bf \hat {x}^{-}_{k}})+{\bf v_k}] [{\bf  h}({\bf {x}_{k}}-{\bf \hat {x}^{-}_{k}})+{\bf v_k})]^\top\\
&= E[{\bf h}({\bf \Delta \hat {x}^{-}_{k}})+{\bf v_k}] [{\bf  h}({\bf \Delta \hat {x}^{-}_{k}})+{\bf v_k})]^\top\\
&= {\bf  h}E[({\bf \Delta \hat {x}^{-}_{k}}) ({\bf \Delta \hat {x}^{-}_{k}})^\top] {\bf  h^\top} + E[ {\bf v_k} {\bf v}^\top_k] \\
&= {\bf h}{\bf \hat P_{k}} {\bf  h^\top} + {\bf  R_{k}} \\
&={\bf P_{zz,k}}
  \end{aligned}
\end{equation}

Thus, auto-covariance of innovation covariance is equal to 
 \begin{equation}
 \begin{split}
{\bf P_{zz,k}}&=E[{\bf \upsilon_{k}} {\bf \upsilon_{k}^\top}]\\
&=\sum\limits_{i=1}^{2L} {({\bf Z_{k-1}})_i} {({\bf Z_{k-1}})_i^\top}-{\bf \hat z_{k-1}}^{-} {\bf \hat z_{k-1}}^{{-}^\top}+{\bf R_{k}^{*}}\\
{\bf \hat R}_{k}^{*}&=E [{\bf \upsilon_{k}} {\bf \upsilon_{k}^\top}]-{\bf h_{k-j}}{\bf P_{k-j}} {\bf h_{k-j}^\top}
  \label{stateSpaceForm1}
  \end{split}
  \end{equation}

After applying  WMAM, uses sample sequence of innovation as 
   \begin{equation}
{\bf \hat R}_{k-j}^{*}=\frac {1}{N_w}\sum_{j=j0}^{k} {\omega({\bf \sigma^2_{k-j}})}{\bf \upsilon_{j}}{\bf \upsilon_{j}^\top}-{\bf h_{k-j}}{\bf P_{k-j}} {\bf h_{k-j}^\top}
  \label{stateSpaceForm1}
  \end{equation}

Similarly, for estimating the process noise covariance matrix, we consider both innovation and residual vector. The residual vector is 
 \begin{equation}
 {\bf \eta_{k}}={\bf z_{k}}-{\bf \hat z_{k}}
  \label{eq31}
  \end{equation}
Taking sample mean and weighed covariance  of the residual sequence is approximated by 
 
  \begin{equation}
  \label{eq:t}
  \begin{aligned}
{\bf\bar \eta_{k}}&=\frac{1}{{N_w}}\sum\limits_{j=1}^{{N_w}} {\bf \eta_{k-j}},\\
E[{\bf \eta_{k}} {\bf \eta_{k}}^\top]&= \frac{1}{{N_w}-1}\sum\limits_{j=1}^{{N_w}} {\bf \omega({\bf \hat\sigma^2_{k-j}})} ({\bf \eta_{k-j}}-{\bf\bar \upsilon_{k}}) ({\bf \eta_{k-j}}-{\bf\bar \eta_{k}}^\top)
  \end{aligned}
\end{equation} 
From equation (20) and equation (32), the difference error is   

 \begin{equation}
  \label{eq:t}
  \begin{aligned}
 ({\bf \upsilon_{k}}- {\bf \eta_{k}}) &= {\bf h(.)}({\bf {x}_{k}}-{\bf \hat {x}^{-}_{k}})\\
 ( {\bf \upsilon_{k}}- {\bf \eta_{k}})  ( {\bf \upsilon_{k}}- {\bf \eta_{k}})^\top &= [{\bf h(.)}({\bf {x}_{k}}-{\bf \hat {x}^{-}_{k}})] [{\bf h(.)} ({\bf {x}_{k}}-{\bf \hat {x}^{-}_{k}})]^\top
  \end{aligned}
\end{equation}

Take expectation
 \begin{equation}
  \label{eq:t}
  \begin{aligned}
E[( {\bf \upsilon_{k}}- {\bf \eta_{k}})  ( {\bf \upsilon_{k}}- {\bf \eta_{k}})^\top]&=E[{\bf \upsilon_{k}} {\bf \upsilon_{k}}^\top]- E[{\bf \upsilon_{k}} {\bf \eta_{k}}^\top] - E[{\bf \eta_{k}} {\bf \upsilon_{k}}^\top] + E[{\bf \eta_{k}} {\bf \eta_{k}}^\top]\\
&= E[{\bf \upsilon_{k}} {\bf \upsilon_{k}}^\top]+ E[{\bf \eta_{k}} {\bf \eta_{k}}^\top]\\
&= {\bf  h}E[({\bf \Delta \hat {x}^{-}_{k}}) ({\bf \Delta \hat {x}^{-}_{k}})^\top] {\bf  h^\top} + {\bf  h}E[({\bf \Delta \hat {x}_{k}}) ({\bf \Delta \hat {x}_{k}})^\top] {\bf  h^\top} \\
&= {\bf h}({\bf \hat P_{k}} + {\bf \hat P_{k}^{-}}){\bf  h^\top}  \\
 \end{aligned}
\end{equation}
Inserting equation (7), (17) into equation (34) and by solving the matrix equation, the estimated process noise covariance matrix can be obtained as
  \begin{gather}
{\bf \hat Q}_{k-j}^{*} =\frac{1}{{N_w}}[\sum\limits_{j=0}^{N} {\omega({\bf \sigma^2_{k-j}})} {\bf \upsilon_{j}}  {\bf \upsilon_{j}}^\top+\sum\limits_{j=0}^{N} {\omega({\bf \sigma^2_{k-j}})} {\bf \eta_{j}}  {\bf \eta_{j}}^\top ]\\ \notag -{\bf h_{k-j}} [\frac{1}{2L}\sum\limits_{i=1}^{2L}{({\bf X_{k-1}})_i} {({\bf X_{k-1}})_i^\top} -{\bf \hat x_{k-1}}{\bf \hat x_{k-1}}^{{-}^\top}+  {\bf \hat P_{k-j}} ] {\bf h_{k-j}^\top}
\end{gather}

\end{proof}

In the CMRACKF algorithm, measurement updated equations are adapted with measurement noise covariance matrix. In this step, filter gain and the auto covariance of innovation sequence is updated by ${\bf R}_{k}^{*}$. Similarly, ${\bf Q}_{k}^{*}$ is calculated based on residual vector as shown in equation (\ref{eq31}) the limited proof due to page limit of the letter. 

In the proposed CMRACK, the ${\bf P_{zz,k}}$ is the auto covariance of innovation sequence updated as
\begin{equation} 
{\bf P_{zz,k}}=\sum\limits_{i=0}^{2L}{\bf W_{i}^{c}}[{({\bf Z_k})_i}-{\bf \hat z_{k}}^{-}] [ {({\bf Z_k})_i}-{\bf \hat z_{k}}^{-}]^\top+{\bf R}_{k}^{*}
 \label{stateSpaceForm1}
 \end{equation}

% To evaluate the Kalman gain and the updated predcited and measurements updated in WCMACKF algorithms the same with CKF. 
% \begin{equation} 
%   \label{eq:t}
%   \begin{aligned}
% {\bf K_{k}}&={\bf P_{xz,k}}{\bf P_{zz,k}^{-1}},\\
% {\bf \hat x_{k}}& ={\bf \hat x_{k}^{-}}+{\bf K_{k}}({\bf z_{k}}-{\bf \hat z_{k}^{-}})\\
% {\bf \hat P_{k}} &={\bf \hat P_{k}}^{-}-{\bf K_{k}}{\bf P_{zz,k}}{\bf K_{k}^\top}.
% \end{aligned}
%  \end{equation}
 
where, $ {\bf \hat R}_{k}^{*}$, $ {\bf \hat Q}_{k}^{*}$, is updated measurement noise covariance matrix. Authors are shown the $ {\bf \hat R}_{k}^{*}$ adaption only in proposed approach.

\section{Numerical Simulation} 
In order to show the effectiveness of the proposed adaptive algorithm, target track example is considered \cite{gao2015sage}. The linear system and nonlinear measurement model of target tracking example can be expressed as follows \cite{wan2000unscented}:

\begin{equation}
      \label{eqn4.3}
    \begin{cases}
    {\bf x_{1}}(k+1)&={\bf x_{1}}(k)+T_s{\bf x_{3}}(k)+{\bf w_{1,k-1}}\\

    {\bf x_{2}}(k+1)&={\bf x_{2}}(k)+T_s( -{\bf k_{x}} {\bf x_{3}}^2)(k)+{\bf w_{2,k-1}}\\
     {\bf x_{3}}(k+1)&={\bf x_{3}}(k)+T_s{\bf x_{4}}(k)+{\bf w_{3,k-1}}\\   
    
    {\bf x_{4}}(k+1)&={\bf x_{4}}(k)+T_s( {\bf k_{y}} {\bf x_{3}}^2-g)(k)+{\bf w_{4,k-1}}\\
     \end{cases}
    \end{equation}
Where, the state, $x_{k}=[ x_{1,k} \quad x_{2,k} \quad x_{3,k} \quad x_{4,k}]$ are the vehicle position and velocity in x and y-plane and its constant coefficients. $T_s$ = 0.1 s, the step size.

% The radar is located at $(s_{x}, s_{y})$, it measures the range ($\gamma_k$) and bearing angle $(\theta_{k})$ on time epoch k. Thus 

The measurement model is represented as

\begin{equation}
      \label{eqn4.3}
    \begin{cases}
      {\bf z_{1}}(k)&= \sqrt{(\bf x_{1}({k})-s_{x})^2+(\bf x_{3}(k)-s_{y})^2} +   {\bf v_{1}}(k)  \\  {\bf z_{2}}(k)&= {\tan^{-1}( \frac{{\bf x_{3}}(k)-s_{y} }{ {\bf x_{1}}(k)-s_{x}}) } + {\bf v_{2}(k)} \\     \end{cases}
    \end{equation}

% It assumes that the target is moving linearly and making a turning movements. The target conducts a linear acceleration curvilinear motion during 0-100s, variable acceleration. The initial coordinate of the target is (0m, 500m), initial velocity (50 m/s, 0 m/s).
 Initial process ($Q$) and measurement noise covariance matrices($R$) are selected manually using trail and error method during simulations. The proposed algorithm is applied for checking the performance in  a target tracking example.  The positions of true trajectory and estimated target trajectory is shown in Figure 1. 

\begin{figure}[!ht]
\centering
\includegraphics[scale=0.55]{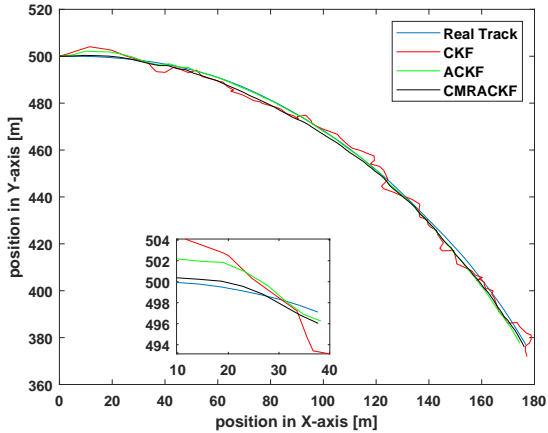}
 \caption{Target tracking result for position and its estimation.}
\label{fig:2}
\end{figure}

\begin{figure}[!ht]
\centering
\includegraphics[scale=0.55]{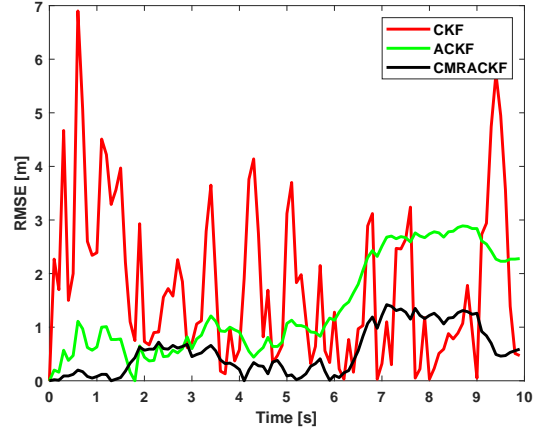}
 \caption{Position RMSE of considered algorithms.}
\label{fig:2}
\end{figure}

Figure 1 depicts the tracking accuracy with CKF, ACKF, proposed algorithm, it can be observed that  CKF can not track the true state ($x_{1,k}$). Nevertheless, the ACKF (green) and proposed algorithm (black) can track position $x_{1,k}$ accurately shown in zoomed version too.

% \begin{table}[ht]
% \caption{Average RMSE of considered algorithms }
% \label{tab:1}  
% \begin{tabular}{lllll}
% \hline\noalign{\smallskip}
%  & RMSE[m] \\
% %  &  RMSE &  RMSE\\
% %   & ${ \degree\per\second}$  & ${\meter\per\second^{2} }$ \\
% \noalign{\smallskip}\hline\noalign{\smallskip}
%  CKF &  $ 5.10\times10^{-2}$ \\
% ACKF &  $ 5.10\times10^{-2}$\\
% WCMACKF &  $1.50\times10^{-3}$\\
% \noalign{\smallskip}\hline
% \end{tabular}
% \end{table}

% \begin{figure}[!ht]
% \centering
% \includegraphics[scale=0.55]{Speed.eps}
%  \caption{Target tracking result for position and its estimation.}
% \label{fig:2}
% \end{figure}

% \begin{figure}[!ht]
% \centering
% \includegraphics[scale=0.55]{RMSE_Speed.eps}
%  \caption{Position RMSE of considered algorithms.}
% \label{fig:2}
% \end{figure}

Figure 2 shows the position RMSE error of CKF, ACKF and proposed method for target tracking. It can be observed that the CKF is having larger estimation error owing to fixed values of noise covariance, $Q$ and $R$. The average RMSE values of CKF, ACKF and the proposed algorithms are calculated as $1.79 m$, $1.40 m$  and $0.55 m$, respectively.  However, CMRACKF algorithm outperforms the CKF and ACKF.  In conclusion, the CMRACKF algorithm has better tracking ability.

\section{Conclusions}
The summary of the paper presents a novel CMRACKF algorithm.  In the proposed algorithm, moving window method based on covariance matching is applied to  evaluate the the measurement noise co-variance matrix that can adaptively adjust statistical noise parameters on-line. The noise covariance matrices is used, then feedback to the standard ACKF to overcome the limitation of ACKF. Numerical simulation results reveal that the  covariance matching framework RACKF can reduce the RMSE by  10\%  approximately and it can enhance the adaptive capability of ACKF algorithm.

\bibliographystyle{ieeetr}
\bibliography{Reference}

% that's all folks
\end{document}